\documentclass[12pt]{article} 
\title{Tight Paths and Tight Pairs \\ in Weighted Directed Graphs%
}  
\author{Jos\'e Luis Balc\'azar \\
{Department of Computer Science} \\
{Universitat Polit\`ecnica de Catalunya}}
\usepackage{fullpage}
\usepackage{hyperref}
\usepackage{algorithm2e}
\SetKwComment{Comment}{// }{}
\usepackage{graphicx}
\newtheorem{theorem}{Theorem}
\newtheorem{lemma}[theorem]{Lemma}

\newtheorem{proposition}[theorem]{Proposition}
\newtheorem{definition}[theorem]{Definition}
\newtheorem{example}[theorem]{Example}
\newenvironment{proof}{\smallskip\noindent{\bf Proof.} \ignorespaces}{\nobreak\hglue0pt\nobreak\hfill\vrule depth 0ex height 1ex width 1ex\smallskip}
\def\st{\bigm|}
\def\g{\gamma}
\def\implies{\Rightarrow}

\begin{document}

\maketitle

\begin{abstract}
We state the graph-theoretic computational problem of 
finding tight paths in a directed, edge-weighted graph, 
as well as its simplification of finding tight pairs. 
These problems are motivated by the need of algorithms 
that find so-called basic antecedents in closure spaces, 
in one specific approach to data analysis.
We discuss and compare several algorithms to approach these 
problems.
\end{abstract}


\section{The general problem}

Consider a finite directed graph $G=(V,E,d)$ on 
a set $V$ of vertices, with edges $E\subseteq V\times V$ 
and positive
costs (or distances) $d:E\to I\!\!R^+$ on the edges.
A \emph{path} in $G$ is a sequence of vertices, $p = (u_1,\ldots,u_k)$, 
with $1 \leq k$ and $(u_i,u_{i+1})\in E$ for $1\leq i < k$. 
The \emph{length} of the path is $k$.
The extreme case of a single node, $k = 1$, is still considered
a (very short) path of length 1. We call $p' = (u_i,\ldots,u_j)$ with
$1 \leq i \leq j \leq k$ a \emph{subpath} of $p$. It is a
\emph{proper subpath} if, besides, $p' \neq p$, that is, either
$1 < i$ or $j < k$ (or both). In general, paths
may not be necessarily simple, that is, repeated passes through
the same vertex are allowed. The cost of a path $p = (u_1,\ldots,u_k)$ 
is the sum $d(p) = \sum_{1\leq i < k} d(u_i,u_{i+1})$ of the costs 
of its edges. Clearly, $d(p) = 0$ if $p$ has length 1.

\begin{definition}
\label{d:tightpath}
Given a graph $G$ and a threshold value $\g\in I\!\!R$ with $\g\geq0$, 
a \emph{tight path} $(u_1,\ldots,u_k)$, is a path $p$ in $G$ such that
\begin{enumerate}
\item
$d(p) \leq\g$ but
\item
every extension of $p$ into $p'$ by adding a new edge at either end, 
$(u_0,u_1)\in E$ or $(u_k,u_{k+1})\in E$,
leads to $d(p')>\g$. (Of course, this might hold vacuously if
no such extension is possible.)
\end{enumerate}
\end{definition}

\begin{example}
\label{ex:simple}
Consider $G$ on $V = \{ A, B, C, D, E \}$ with edges and
costs $(A,B): 2$, $(B,C): 1$, $(C,E): 1$, $(A,D): 1$ and
$(D,E): 2$ (see Figure~\ref{fig:simple}). For threshold 3, the tight paths are $(A,B,C)$,
$(B,C,E)$ and $(A,D,E)$.
\end{example}

\begin{figure}
    \centering
    \includegraphics[width=0.6\linewidth]{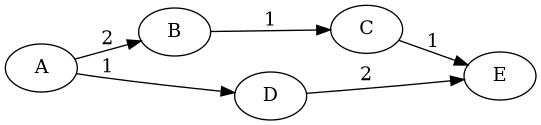}
    \caption{Graph for example \ref{ex:simple}.}
    \label{fig:simple}
\end{figure}

\begin{example}
\label{ex:withcycles}
Consider $G$ on $V = \{ A, B, C, D, E \}$ with edges and
costs $(A,B): 2$, $(B,C): 1$, $(C,A): 1$, $(A,D): 1$,
$(D,E): 2$, and $(E,A): 1$ (see Figure~\ref{fig:withcycles}). 
For threshold 4, the tight paths are 
$(A,B,C,A)$, $(A,D,E,A)$, $(B,C,A,B)$, $(B,C,A,D)$, 
$(C,A,B,C)$, $(C,A,D,E)$, $(D,E,A,D)$, $(E,A,D,E)$, 
and $(E,A,B,C)$. At threshold 5, we find $(A,B,C,A,D)$, 
$(D,E,$ $A,B)$, and others (a total of 9, all of
length 5 except one). For larger thresholds, tight paths
can go through the complete loops twice or more, or also 
alternate between them repeatedly.
\end{example}

\begin{figure}
    \centering
    \includegraphics[width=0.4\linewidth]{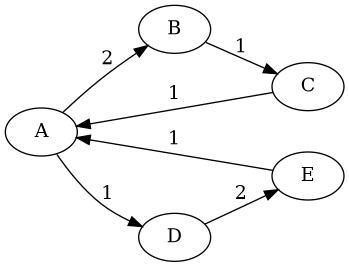}
    \caption{Graph for example \ref{ex:withcycles}.}
    \label{fig:withcycles}
\end{figure}

\begin{proposition}
Every vertex belongs to at least one tight path.
\label{p:vcover}
\end{proposition}

\begin{proof}
For any vertex $u\in V$, we grow a tight path edge-wise. 
View initially $u$ as a path of length 1, hence, cost $0\leq\g$. 
As long as the path at hand is not tight, it can be extended 
into a longer one with cost still bounded by $\g$, never losing 
$u$ along the way, until tightness is reached.
\end{proof}

An alternative formulation of the same intuition of a tight path
could be to consider subpaths of paths longer by more than one 
additional edge: a proper subpath of a tight path would 
not be tight. Both approaches are equivalent:

\begin{proposition}
Given graph $G$ with positive edge costs
and threshold $\g\in I\!\!R$ with $\g\geq0$, 
a path $p$ in $G$ is tight if and only if
$d(p) \leq\g$ and
$d(p') >\g$ for every $p'$ that has $p$ as a proper subpath.
\label{p:extendpath}
\end{proposition}

\begin{proof}
This is easily argued. Assume that $p$ is tight: if $p$ is a proper 
subpath of $p'$, at least one additional edge is present in $p'$;
moreover, either one such additional edge comes after the last vertex 
of $p$, or before its first vertex. 
Then, as $p$ is tight, the cost of $p$ plus the cost of that extra edge 
is already beyond the cutoff and the cost of $p'$ is at least that sum.
The converse is proved by simply considering $p'$ that extends
$p$ by one edge at the beginning or at the end. Note that
the condition that costs are positive is relevant to this argument.
\end{proof}


This paper focuses on the problem: given $G$ and $\g$ as indicated, 
find all the tight paths. We consider also a variant motivated by 
earlier work in the data analysis field, discussed next.

\section{A case of acyclic vertex-weighted graphs}
\label{s:acyclic}

In the context of a study of how partial implications can be 
obtained from closure spaces associated to transactional
datasets~\cite{Bal2010LMCS,BalTir2011EGC}, a close relative 
of tight paths arises; actually, an interesting particular 
case. There, the graph is naturally acyclic, weights actually 
apply to vertices, and edge costs are obtained from the 
vertices' weights (we develop a bit more this motivating 
connection below).

More precisely, we focus now on the following particular case:
our input graph is $G=(V,E,w)$, with vertex weights $w:V\to I\!\!R^+$. 
Moreover, the following inequality is assumed to hold: 
for all $(u,v)\in E$, $w(u) < w(v)$; {\em a fortiori}, 
$G$ is acyclic. 
Then, costs are set on the edges as weight differences
between the endpoints, $d((u,v)) = w(v) - w(u)$ for
every edge $(u, v)$, and tight paths are to be found 
on this graph for a cutoff value provided separately 
as before. Note that, by the assumed inequality about 
weights in edge endpoints, all these costs will be 
strictly positive, hence we fulfill the conditions 
of the general case.

However, now, tight paths correspond to tight ``vertex pairs'' 
in the sense that, for a pair $(u,v)$, the weights of all the 
paths from $u$ to $v$ coincide, and depend only on the weights 
of the endpoints. Thus, we speak of ``tight pairs'' instead, as 
endpoints of tight paths. That fact is easy to see but we report 
here a proof for the sake of completeness.

\begin{lemma}
Given graph $G$ with vertex weights $w$ where
$w(u) < w(v)$ for all $(u,v)\in E$ and corresponding
edge costs $d((u,v)) = w(v) - w(u)$, the cost of a
path is $d(u,u_2,\ldots,u_{k-1},v) = w(v) - w(u)$.
\label{l:acyclicpathsum}
\end{lemma}

\begin{proof}
We prove it by induction on the length of the path.
For length 1, $u = v$ and indeed the cost is $w(u) - w(u) = 0$.
Assume $k > 1$, and 
$d(u,u_2,\ldots,u_{k-1}) = w(u_{k-1}) - w(u)$
as inductive hypothesis. Then $d(u,u_2,\ldots,u_{k-1},v)$ is:
$$
d(u,u_2,\ldots,u_{k-1}) + d((u_{k-1},v)) = 
w(u_{k-1}) - w(u) + w(v) - w(u_{k-1}) = 
w(v) - w(u).
$$
Note that, for $k = 2$, that is, for a path consisting of a 
single edge, this boils down exactly to how the cost of the edge
has been defined. 
\end{proof}

\begin{example}
\label{ex:acyclic13}
The graph in Figure~\ref{fig:acyclic13} comes from a closure space
associated to a toy dataset. The construction of the closure space
provides us (after a minimal preprocessing) with the labels shown 
in the vertices. We then compute the differences to obtain the costs
for the edges and reach the particular case of a vertex-weighted acyclic
graph discussed in this section; edge costs are shown only for clarity. 
It is a simple matter to check out some of the equalities stated
by Lemma~\ref{l:acyclicpathsum} in this example. We will explain 
later how to work directly with the vertex labels instead,
in algorithmic terms, avoiding the edge-labeling process.
\end{example}

\begin{figure}
    \centering
    \includegraphics[width=0.8\linewidth]{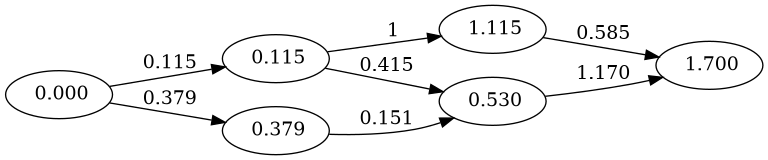}
    \caption{Graph for example \ref{ex:acyclic13}.}
    \label{fig:acyclic13}
\end{figure}

\subsection{Correspondence tightening}
\label{ss:ct}

Our next aim is to describe an algorithm to find tight pairs 
in the particular case of this section~\ref{s:acyclic}.
It has not been published, to our knowledge, in standard scientific 
literature but has been publicly available for quite some time as code 
in the implementation of certain operations on closure lattices in a
somewhat old 
open-source repository\footnote{See \url{https://github.com/balqui/slatt}, 
contributed by the current author; however, that repository is 
almost abandoned and a pointer to a more recent reimplementation will be 
given at the end of Section~\ref{ss:algct}.}.
It is not described in graph-theoretic terms; in fact, a main purpose
of this paper is to show in later sections the advantages of a 
graph-theoretic approach.
Instead, it is developed in terms of tightening of binary relations 
(a.k.a.~correspondences):
our relational setting emphasizes a more general view that departs
from distances and path weights.

We identify algorithmically useful properties
in two steps. First, we introduce a very general notion of tightening
for correspondences among ordered sets and prove some properties
that do not need additional assumptions. Then, we identify a 
couple of extra relevant assumptions that lead to further 
properties, so as to reach the desired algorithm.

Recall that a correspondence, or binary relation, is just
a subset of a cartesian product $A\times B$ of two sets
$A$ and $B$. We consider only finite sets.
Given $R\subseteq A\times B$ where both $A$ and $B$ are
partially ordered sets, we can define the \emph{tightening} of $R$
as another correspondence among the same sets, as follows: 
$$
t(R) = \{ (a,b)\in R \st 
\forall a'\in A \, \forall b'\in B \,
(a' \leq_{A} a \land b \leq_{B} b' \land (a', b')\in R \implies a' = a \land b' = b
\}
$$
That is, the tightening keeps those pairs of related elements $(a,b)\in R$
such that moving down from $a$ in $A$ and/or up from $b$ in $B$, changing at
least one of them, implies losing the relationship. 

Clearly $t(R) \subseteq R$.
We say that $R$ is \emph{tight} if 
$t(R) = R$. Of course, tightening again a tight relation does not change it:

\begin{proposition}
\label{p:idempotent}
$t(t(R)) = t(R)$.
\end{proposition}

\begin{proof}
Since $t(R)\subseteq R$ for all $R$, we have that $t(t(R))\subseteq t(R)$.
Let $(a, b)\in t(R)$, and consider any $(a', b')\in t(R)$ with 
$a' \leq_{A} a$ and $b \leq_{B} b'$: then, by $t(R)\subseteq R$ again,
$(a', b')\in R$. Now, application of the definition of $t(R)$ leads 
to $a' = a$ and $b' = b$, hence $(a, b)\in t(t(R))$.
\end{proof}

To motivate the consideration of this operator, we explain first how it connects
with (the acyclic, vertex-weighted case of) the tight pairs problem.
Indeed, we can implement a correspondence of the following form:

\begin{definition}
\label{d:RGg}
For a graph $G=(V, E, w)$, with $w:V\to I\!\!R^+$ as
in the previous section, and a cutoff value $\g$,
$R_{G,\g}\subseteq V\times V$ consists of the pairs $(u, v)$ 
of vertices joined by a path in $G$ from $u$ to $v$ 
and such that $w(v) - w(u) \leq \g$.
\end{definition}

Recall from Lemma \ref{l:acyclicpathsum} that 
$w(v) - w(u)$ is the total distance along the path
from $u$ to $v$. The following holds:

\begin{theorem}
\label{t:tistight}
The tight pairs in $G$ coincide with $t(R_{G,\g})$, 
where the partial ordering at both sides is defined 
by the reflexive and transitive closure of the acyclic 
graph $G$ itself.
\end{theorem}

\begin{proof}
For $G = (V, E, w)$, denote $E^*$ the reflexive and transitive
closure of the edge relation. Due to acyclicity, it defines
a partial order $\leq_E$.
Let us consider a tight pair in $G$, say $(u, v)$, with $u\leq_E v$. 
This means that there is a path in $G$ from $u$ to $v$; further, 
$w(v) - w(u) \leq \g$, which implies that $(u, v) \in R_{G,\g}$. 
By Proposition~\ref{p:extendpath}, an extended path from
$u'$ to $u$ then to $v$ then to $v'$ incurs in a distance 
$w(v') - w(u') > \g$ unless it coincides with $(u, v)$,
that is, if $u'\leq_E u \land v\leq_E v' \land (u', v') \in R_{G,\g}$
then $u' = u$ and $v = v'$. The argumentation works both ways.
\end{proof}

\subsection{Algorithm for correspondence tightening}
\label{ss:algct}

Given the graph and the threshold, we obtain tight pairs
in three phases. The first phase constructs $R_{G, \g}$
and the next two phases reach the tightening, left-hand side 
first, then right-hand side. Given that the right-hand process 
is applied on the outcome of the previous phase, and not on the literal 
correspondence obtained from the graph, we must argue
separately why the result is actually correct.

Let us clarify these one-sided versions of tightening first.
Consider $R\subseteq A\times B$ and assume that $B$ is equipped 
with a partial order. Then, we can define the \emph{$r$-tightening} of a
correspondence as another correspondence among the same
sets, as follows: for $R\subseteq A\times B$, $t_r(R)$ 
contains the pairs $(a,b)\in R$ such that $b$ is maximal 
in $\{ b\in B \st (a,b)\in R\}$. Formally, $(a,b)\in t_r(R)$
when $(a,b)\in R$ and
$$
\forall b'\in B \, (b \leq_{B} b' \land (a, b')\in R \implies b' = b).
$$

Similarly, when $A$ has a partial order, 
one can propose the dual definition of $\ell$-tightening on the
first component in the natural way: $t_{\ell}(R)$ contains 
the pairs $(a,b)\in R$ such that $a$ is minimal in 
$\{ a\in A \st (a,b)\in R\}$. Again, formally, 
$(a,b)\in t_{\ell}(R)$ when $(a,b)\in R$ and
$$
\forall a'\in A \, (a' \leq_{A} a \land (a', b)\in R \implies a' = a).
$$
We can consider pairs $(a,b)\in R$ where both $a$ is minimal in
the set $\{ a\in A \st (a,b)\in R\}$ and $b$ is maximal in 
$\{ b\in B \st (a,b)\in R\}$, that is, requiring simultaneously
the conditions of $t_r$ and of $t_{\ell}$. How does $t$ compare to it?

\begin{proposition}
\label{p:lateralt}
\begin{enumerate}
\item
$t_{r}(R) \cup t_{\ell}(R) \subseteq R$.
\item
$t(R) \subseteq t_{r}(R) \cap t_{\ell}(R)$.
\item
If $t(R) = R$ then $t_{\ell}(R) = R = t_{r}(R)$.
\end{enumerate}
\end{proposition}

\begin{proof}
\begin{enumerate}
\item
By definition, in the same way as $t(R) \subseteq R$.
\item
Taking $a' = a$ in the definition of $t(R)$ proves inclusion in
$t_r(R)$ and, likewise, taking $b' = b$ completes the proof.
\item
By the previous two parts, we get sandwiched between $t(R)$ and 
$R$ all of 
$t_{r}(R) \cap t_{\ell}(R)$, $t_{r}(R)$, $t_{\ell}(R)$ and 
$t_{r}(R) \cup t_{\ell}(R)$.
If, besides, $R \subseteq t(R)$, then all these sets must coincide.
\end{enumerate}
\end{proof}

We obtain further mileage when both partial orders are the same and $R$ is
a subrelation of the partial order having a sort of ``convexity'' property:

\begin{definition}
\label{d:convex}
Let $(A, \leq)$ be a partially ordered set and $R\subseteq A\times A$;
$R$ is \emph{convex} (with respect to $\leq$) if 
$(a, b)\in R \implies a \leq b$ and whenever 
$a \leq b \leq c$ and $(a, c)\in R$ then $(a, b)\in R$ and 
$(b, c)\in R$ too. 
\end{definition}

Then, we can strengthen Proposition~\ref{p:lateralt} further:

\begin{proposition}
\label{p:lateraltstrong}
If $R\subseteq A\times A$ is convex with respect to 
a partial order $(A, \leq)$ then:
\begin{enumerate}
\item
$t(R) = t_{r}(R) \cap t_{\ell}(R)$.
\item
$t(R) = R$ if and only if $t_{\ell}(R) = R = t_{r}(R)$.
\end{enumerate}
\end{proposition}


\begin{proof}
\begin{enumerate}
\item
To see the inclusion that is not covered by Proposition~\ref{p:lateralt},
assume $(a, b) \in t_{r}(R) \cap t_{\ell}(R)$, that is:
$$
\forall a'\in A \, (a' \leq_{A} a \land (a', b)\in R \implies a' = a)
$$
and
$$
\forall b'\in B \, (b \leq_{B} b' \land (a, b')\in R \implies b' = b)
$$
where now, actually, $A = B$ and $\leq_A$ is the same as $\leq_B$.
Suppose that $(a'', b'')\in R$ with $a'' \leq a \leq b \leq b''$.
Then, by convexity, both $(a'', b)\in R$ and $(a, b'')\in R$,
and we can infer $a'' = a$ and $b'' = b$, whence $(a, b) \in t(R)$.
\item
To see the implication that is not covered by Proposition~\ref{p:lateralt},
assume $t_{\ell}(R) = R = t_{r}(R)$: then, by the previous part,
$t(R) = t_{r}(R) \cap t_{\ell}(R) = R \cap R = R$.
\end{enumerate}
\end{proof}

Some condition akin to convexity is necessary for this proposition to
go through, as it fails in general. For example, consider 
$V = \{1, 2, 3, 4\}$ with its standard partial order (that is 
actually total). We choose $R = \{ (1, 2), (1, 4), (2, 3), (3, 4) \}$
(see figure \ref{fig:nonconvex}), which is easily seen not
to be convex because $(1, 4)$ is present but $(1, 3)$ is missing,
as is $(2, 4)$. There, we can see that the intersection property fails:
$t(R) = \{(1, 4)\}$ but 
$t_{r}(R) = \{(1, 4), (2, 3), (3, 4)\}$,
$t_{\ell}(R) = \{(1, 2), (1, 4), (2, 3)\}$, and
$t_{r}(R) \cap t_{\ell}(R) = \{(1, 4), (2, 3)\}$.

\begin{figure}
    \centering
    \includegraphics[width=0.50\linewidth]{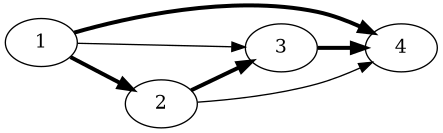}
    \caption{A non-convex relation, marked by wider edges, on the total order $\{1, 2, 3, 4\}$.}
    \label{fig:nonconvex}
\end{figure}

Of course, the convexity property holds for the relation 
$R_{G,\g}$ with respect to the ordering given by the reflexive and 
transitive closure of the edges of $G$ because, if
the distance from $a$ to $c$ is bounded by $\g$, then so are the
distances among any intermediate vertices along any path from $a$ to $c$.
Thus, we have all these properties available to 
apply our operators in order to find tight pairs of $G$.
Suppose that we are given $R$ by an explicit representation of the 
whole partial order (that is, the explicit list of all pairs in the 
reflexive and transitive closure of the set of edges $E$ in the
acyclic graph case)
plus a criterion to know whether a pair $a\leq b$ is actually in $R$
(the $\g$ bound in our case);
then, we can tighten $R$ (that is, compute an explicit representation 
of $t(R)$) by means of an algorithm that runs three phases:

\begin{enumerate}
\item 
Take from the partial order the pairs that belong to $R$ and discard 
the rest.
\item
Tighten $R$ at the left: once we have $R$, apply $t_{\ell}$ to obtain $t_{\ell}(R)$.
\item
Tighten at the right: apply $t_r$ to obtain $t_r(t_{\ell}(R))$.
\end{enumerate}

Then, we find tight pairs by feeding $R_{G,\g}$ into this algorithm;
of course, to prove it correct, we need to argue the following:

\begin{theorem}
Let $R\subseteq A\times A$ with $(A, \leq)$ a partial order. If 
$R$ is convex with respect to it, 
then, $t(R) = t_{r}(t_{\ell}(R))$; in particular, given $G$ and
$\g$, $t_{r}(t_{\ell}(R_{G,\g}))$ is the set of tight pairs of $G$
at threshold $\g$.
\end{theorem}

\begin{proof}
We first rewrite membership in $t_{r}(t_{\ell}(R))$ as follows:
by the definition of $t_r$,
$$
(a, b) \in t_{r}(t_{\ell}(R)) \iff 
(a,b)\in t_{\ell}(R) \land 
\forall b'\in B \, 
(b \leq_{B} b' \land (a, b')\in t_{\ell}(R) \implies b' = b).
$$
Now, by the definition of $t_{\ell}$, 
and substituting a disjunction with negated left-hand for 
the last implication,
$(a, b) \in t_{r}(t_{\ell}(R))$ if and only if $(a,b)\in R$ 
and
$$\forall a'\in A \, 
(a' \leq_{A} a \land (a', b)\in R \implies a' = a)
\land {}
$$
\removelastskip
$$
\forall b'\in B \, 
(\lnot b \leq_{B} b' \lor (a, b')\notin t_{\ell}(R) \lor b' = b).
$$
Finally, replacing the last occurrence of $t_{\ell}$ by its
definition while taking into account that $(a,b)\in R$ has been
already stated, $(a, b) \in t_{r}(t_{\ell}(R))$ if and only if
\begin{enumerate}
\item 
$(a,b)\in R \land {}$ 
\item
$\forall a'\in A \, (a' \leq_{A} a \land (a', b)\in R \implies a' = a) \land {}$
\item
$\forall b'\in B$:
\begin{enumerate}
\item 
$\lnot b \leq_{B} b' \lor {}$
\item 
$(a, b')\notin R \lor {}$
\item 
$\exists a'\in A \, (a' \leq_{A} a \land (a', b')\in R \land a' \neq a) \lor {}$
\item 
$b' = b$
\end{enumerate}
\end{enumerate}

As $R$ is assumed to be convex, we have part 1 of
Proposition~\ref{p:lateraltstrong} available: $(a, b) \in t(R)$ if and 
only if $(a, b) \in t_{r}(R) \cap t_{\ell}(R)$, that is, with the same
transformation of the last implication connective, 
$(a, b) \in t(R)$ if and 
only if 
\begin{enumerate}
\item 
$(a,b)\in R \land {}$ 
\item
$\forall a'\in A \, (a' \leq_{A} a \land (a', b)\in R \implies a' = a) \land {}$
\item 
$\forall b'\in B \, (b \leq_{B} b' \land (a, b')\in R \implies b' = b)$,
which is the same as saying that 
for all $b'\in B$:
\begin{enumerate}
\item 
$\lnot b \leq_{B} b' \lor {}$
\item 
$(a, b')\notin R \lor {}$
\item 
$b' = b$
\end{enumerate}
\end{enumerate}

We see that the only difference is part (3c) of the first enumeration.
So, first, the implication from $(a, b) \in t(R)$ to 
$(a, b) \in t_{r}(t_{\ell}(R))$ holds because this second
statement amounts, simply, to adding a disjunct to one of the clauses.

For the converse, we assume that $(a, b) \notin t(R)$
but $(a, b) \in t_{r}(t_{\ell}(R))$, and we will be able to reach
a contradiction. As parts 1 and 2 coincide, $(a, b) \notin t(R)$
must fail its part 3, so that 
$\exists b'\in B \, (b \leq_{B} b' \land (a, b')\in R \land b' \neq b)$
or, equivalently, all three disjuncts (3a), (3b) and (3c) fail. Then, 
for such a particular $b'$, 
the only way to satisfy part 3 of $(a, b) \in t_{r}(t_{\ell}(R))$ is
via the extra disjunct:
$\exists a'\in A \, (a' \leq_{A} a \land (a', b')\in R \land a' \neq a)$.

Now, for such $a'$, convexity 
tells us that $(a, b) \in R$ implies $a \leq_A b$ and that from
$a' \leq_A a \leq_A b \leq_A b'$ and $(a', b')\in R$ we infer
$(a', b)\in R$. But, then, part 2 implies that $a' = a$ whereas
the way $a'$ was chosen includes the condition $a' \neq a$.
Thus, from our assumption that $(a, b) \notin t(R)$ and 
$(a, b) \in t_{r}(t_{\ell}(R))$ we reach a contradiction,
and the equality in the first part of the statement is proved.
The second sentence follows immediately by Theorem~\ref{t:tistight}.
\end{proof}

The details of implementing these three phases are spelled out in 
Algorithm~\ref{a:slatt}. In it, we denote again by $\leq_E$ 
the partial order $E^*$ defined by the (reflexive and
transitive closure of the) edges $E$ of $G$
(and by $<_E$ its corresponding strict version).
We expect to receive the edges of $E^*$ organized in 
``lists of predecessors'', one list per vertex: 
for each vertex $v$, the list contains all the vertices $u$
for which there is a path from $u$ to $v$. This representation
of the graph is chosen because it fits, in the data analysis 
application alluded to above, how the previous phase 
(the closure miner) is providing its output; it would be
a simple matter to adapt the algorithm to alternative 
representations in other use cases. In that application,
the process that produces the graph provides the vertices 
in a topological order, and the existence of a path connecting 
two vertices can be checked in constant time with respect to the
graph size because, in that case, it reduces to a subset/superset 
relationship.

\begin{algorithm}
\DontPrintSemicolon
\KwData{graph $G$: lists of predecessors per vertex, threshold $\g$}
\KwResult{set containing the tight pairs in $G$ for threshold $\g$}
\For{each vertex $v$}{
  create a new list $R_v$ containing the predecessors $u$ of $v$ 
  found in the input list of~$v$ for which $w(v) - w(u) \leq \g$
}
\Comment{The lists $R_v$ implement now $R_{G, \g}$}
\For{each vertex $v$}{
  remove from $R_v$ the vertices that are not minimal in $R_v$
  according to $\leq_E$
}
\Comment{The lists $R_v$ implement now $t_{\ell}(R_{G, \g})$}
\For{each vertex $v$}{ 
  \For{each vertex $u$ such that $u <_E v$}{
      remove from $R_u$ the vertices that appear in $R_v$ as well
    }}
\Comment{The lists $R_v$ implement now $t_r(t_{\ell}(R_{G, \g}))$}
\Return{the list of all pairs $(u, v)$ where $u$ remains in $R_v$}
\BlankLine
\BlankLine
\caption{Tight pairs via correspondence tightening}
\label{a:slatt}
\end{algorithm}

Thus, $t_{\ell}(R_{G,\g})$ is computed by replacing each list 
of predecessors by a list of left-tight predecessors;
then, $t(R_{G,\g})$ is computed by filtering again
these predecessor lists, $t_{\ell}(R_{G,\g})$, so that
each element in a list remains only in those lists 
corresponding to vertices that are maximal among 
those vertices where the element appears in the list.

\begin{example}
\label{ex:acyclic13revisited}
Let us return to Example~\ref{ex:acyclic13} and the graph in Figure~\ref{fig:acyclic13}. Algorithm~\ref{a:slatt} would receive
it as lists of predecessors: an empty list for vertex 0.000, 
lists consisting only of vertex 0.000 for both vertices 0.115 
and 0.379, lists with all these three vertices for both 1.115
and 0.530, and a list with all these five vertices for the sink
vertex 1.700. We are using the weights to refer to vertices, 
given that, in this particular case, all weights are different. 
For $\g = 0.9$, the first 
phase of the algorithm filters out of the lists those predecessors 
that are too far away: the list for vertex 1.115 becomes empty, 
the same vertex is the only predecessor left for node 1.700, and 
the remaining lists don't lose any element. In the second phase,
the two lists that were reduced stay the same but, 
in all the other lists, 
only vertex 0.000 is minimal and is the only one that remains:
this amounts to left-tightening. Then, in the third phase, we
tighten at the right: vertices 0.115 and 0.379 lose their single
left-tight predecessor 0.000 because it also appears as predecessor
of the later vertex 0.530. Thus, we obtain two tight pairs:
$(0.000, 0.530)$ and $(1.115, 1.700)$. Likewise, assume that 
$\g = 1$ instead. Then, the first phase leaves the same result
as before except that vertex 1.115 keeps its only predecessor as 
the edge with cost 1 is not above the bound, and this predecessor
survives also left and right tightening yielding an additional
tight path $(0.115, 1.115)$. At $\g = 2$, a bound larger than 
all path costs, all predecessors are kept, only 0.000 survives 
left-tightening, only 1.700 survives right-tightening, and only 
one pair $(0.000, 1.700)$ is output, whereas, finally, at $\g = 0.5$ 
the tight paths have all a single edge with cost less than that bound.
\end{example}

An open source Python implementation of the algorithm is 
available in GitHub\footnote{See 
\url{https://github.com/balqui/tightpaths}, file 
\texttt{src/slatt.py}; the graph in Figure~\ref{fig:acyclic13}
is available in file \url{../graphs/fromdata/e13b_logsupp.elist}, 
except that we have reduced the floats' precision in the image for readability.
It is derived from the relational toy dataset 
\url{../graphs/fromdata/e13b.td} also there.}.

\subsection{Analysis}

A glance suffices for a $O(|V|^3)$ time bound for the last loop, that
dominates easy $O(|V|^2)$ time bounds for the rest. For a finer analysis,
the first two loops can be seen as tracing all the predecessors of each 
vertex, for a $O(|E^*|)$ bound, not really smaller unless the order is
far from a total order and has lots of uncomparable vertices. Still they 
are dominated by the third loop where, again, a careful implementation 
will launch the removal operation $O(|E^*|)$ times; the removal itself, 
if the lists are kept sorted, can be implemented by a merge-like simultaneous
traversal, still with $O(|V|)$ in worst case but, if a degree bound is 
available, say $\delta$, we may get the whole runtime down to 
$O(|E^*|\delta)$.

It is not difficult to come up with additional ideas for potential
improvements of the speed of this algorithm; however, we have chosen
to stick to the formulation that has been available for some time in
\url{https://github.com/balqui/slatt}, as explained above. We proceed 
next, instead, to describe a graph-theoretic approach that will provide 
appropriate guarantees of faster processing.

\section{Algorithm for the general case}

One natural option to solve the general problem of finding tight paths
as per Definition~\ref{d:tightpath} is to consider each vertex as a 
candidate to an initial endpoint of one or more tight paths; then, 
an edge traversal starting at that vertex can keep track of the added 
costs and identify the corresponding final endpoints. Note, however, 
that different paths connecting the same pair of vertices, 
but possibly having different weights, may or may not be 
tight, as in Example~\ref{ex:simple}; hence, we need a
traversal of reachable edges, not just of reachable vertices,
and we must accept visiting vertices more than once. The different
paths employed for each visit are to be preserved, and this requires 
somewhat sophisticated bookkeeping.


We loosely follow a depth-first-like strategy 
(see \cite{Sedgewick2011} and/or \cite{Valiente2021})
but, as just argued, 
not of vertices but of edges. 
Strict depth-first traversal of all edges reachable from an initial 
vertex is possible but not very simple, e.g.~the solution proposed in 
the method {\tt edge\_dfs} of the Python package NetworkX maintains a 
stack of half-consumed iterators. Moreover, tight paths may not be
simple paths so repeatedly traversing the same edge may be necessary.
On the other hand, solving our problem does not require us to stick 
to a strict depth-first search order, hence we present a related but 
somewhat simpler algorithm.

Our algorithm is spelled out as Algorithm~\ref{a:tightpaths}. It is to be 
called for each vertex in turn as initial vertex. A bookkeeping tree rooted 
at the current initial vertex (hence called \emph{root}) maintains the 
separate paths, and a stack of tuples is used in a rather standard way. 
For convenience, we assume the availability of a very large cost that 
is always strictly larger than all added costs along the graph and which keeps 
this property when added to any other cost value; we denote it as $\infty$ here, 
and note that the Python programming language offers a {\tt float} value, 
namely {\tt float("inf")}, that has these properties. 

\begin{algorithm}
\DontPrintSemicolon
\KwData{graph $G$, root vertex in $G$, threshold $\g$}
\KwResult{set $R$ containing tight paths in $G$ 
for threshold $\g$ starting at that root}
\BlankLine
let $u$ be the closest predecessor of the root vertex in $G$ if any\;
let $d_0$ be $d(u, \mathrm{root})$ if the current root has some predecessor, 
$\infty$ otherwise\;
initialize $R$ to $\emptyset$\;
initialize an empty tree and add to it a node representing the root vertex\;
initialize an empty stack and push the pair $(\mathrm{root}, 0)$ into it\; 
\BlankLine
\While{\rm stack not empty}{
  let $(v, d)$ be the top or stack, pop it\;
  \For{\rm each $u$ immediate successor of $v$ in $G$}{
    \If{$d + d(v, u) \leq \g$}
      {
      add to tree a new node representing $u$ linked to the tree node $v$\;
      push $(v, d + d(v, u))$ into the stack\;
      }
    }
\If{\rm nothing was pushed into the stack in the \emph{for} loop and $d_0 + d > \g$}
   {add to $R$ the path from the root to $v$ reconstructed from tree (see text)}
}
\BlankLine
\Return{$R$}
\BlankLine
\BlankLine
\caption{Tight paths starting at a root vertex}\label{a:tightpaths}
\end{algorithm}

Distance $d_0$ indicates how to extend paths starting at the root with 
a previous additional vertex, with as short an extra distance as possible.
A value of $\infty$ indicates that no predecessors of the root exist at all. 
Thus, for every candidate to a tight path starting at the root, of cost $d$, 
the condition $d_0 + d > \g$ checks that all potential extensions with a
predecessor of the root lead to exceeding the threshold.

A tuple $(u, d)$ in the stack refers to vertex $u$ having been reached 
from the root vertex
through a path of cost $d \leq \g$. Potential extensions at the $u$ side 
are then considered: if there is any with a new total cost of at most $\g$, 
then the path so far is not yet tight, and all such extensions are added 
to the stack in order to keep building longer paths. Otherwise, the path 
is tight at the $u$ side, and is added to $R$ if it is tight as well at 
the root side, according to $d_0 + d$.

As already indicated, the same vertex in $G$ may be reached in different
ways with different total costs, and it is necessary to distinguish the
corresponding paths. To this end, we implement each vertex in the stack 
not as a pointer to $G$ but as a pointer to the tree. This will allows us 
to reconstruct the path from the tree. In turn, the nodes of the tree refer 
to vertices of the graph, which appear in the tree as many times as different 
paths lead to them. At the time of storing a (surrogate of a) vertex $v$ in the 
tree, as a successor of $u$ (which is also a tree node taken from the stack), 
it receives an additional fresh label to distinguish it from other surrogates 
of $v$, which refer to reaching also $v$ but via other paths. A pointer to 
the parent node $u$ in the tree will allow us to trace the path back, by just 
walking up the tree through these parent pointers.

An open source Python implementation of the algorithm, together with
a Python class \texttt{PathTree} supporting the bookkeeping, is 
available in the same GitHub repository already 
indicated\footnote{See 
\url{https://github.com/balqui/tightpaths}, files 
\texttt{src/tightpath.py} and \texttt{src/pathtree.py};
the graphs in Figures \ref{fig:simple}~and~\ref{fig:withcycles}
are available in files 
\url{../graphs/e1.elist}~and~\url{../graphs/e2.elist}.}.
The reader can look up there the specifics of the implementation of all 
these considerations. Also, a Boolean value \emph{mayextend} is employed 
in that implementation to cover the test about the growth of the stack in 
the last ``if''; here we found it clearer to spell out its meaning instead.

\subsection{Correctness}

To formalize the correctness proof of Algorithm~\ref{a:tightpaths}, 
we develop a series of lemmas. 
All of them refer to a call to the algorithm on graph~$G$, threshold~$\g\geq0$, 
and a root vertex, and to the stack and tree along that run. All proofs regarding 
paths in the graph $G$ or in the tree in the algorithm are constructed by induction 
either on the length of the path or on the iteration of the while loop. Along the 
proofs, we often refer to a tree node $v$ (or its correlate in a pair $(v, d)$ in 
the stack) as a vertex of $G$, since that vertex is uniquely defined.

\begin{lemma}
For each pair $(v, d)$ that enters the stack when the loop just popped $u$
inside the loop, 
\begin{enumerate}
\item 
there is in $G$ a path from the root to $v$ of total cost $d$;
\item
that path has $(u, v)$ as its last edge; 
\item
furthermore, $d\leq\g$ and, from that point 
onward, the tree contains a parent link from $v$ to $u$.
\end{enumerate}
\label{l:pathsTinG}
\end{lemma}

\begin{proof}
Along the loop, let $u$ be just popped from the stack together
with a cost $d'$. If $(v, d)$ enters then the stack,
then $v$ has been found as a successor of $u$ and $d = d' + d(u, v)$; 
inductively, at the time $(u, d')$ entered the stack we had  
a path in $G$ from the root to $u$ of cost $d'$, which extends
through edge $(u, v)$ into a path to $v$ of cost $d = d' + d(u, v)$.
The fact that $d\leq\g$ is checked explicitly before pushing $(v, d)$ 
into the stack, and the parent link is explicitly set at that point.
Links are never removed from the tree.

For the basis of the induction, observe that the path mentioned in 
part 1.~also exists, with length zero, for the pair (root${}, 0$) that 
initializes the stack, since $\g\geq0$, and that parts 2.~and~3.~are
not needed for the inductive argument.
\end{proof}

\begin{lemma}
Let $T$ be the tree at the end of the algorithm.
For each path in $G$ starting at the vertex corresponding
to the root of $T$, of total cost at most $\g$, ending, say, at $v$, 
its reversal appears as a sequence of parent pointers in $T$, from a 
tree node corresponding to $v$ upwards to the root.
\label{l:pathsGinT}
\end{lemma}

\begin{proof}
By induction on the length of the path. The only paths of length 1 starting
at the root vertex both in $T$ and in $G$ correspond to each other.
Given a path of cost at most $\g$ from the root vertex to $v$ in $G$, 
for $v$ different from the root, the path has a last edge $(u, v)$
and the path from the root vertex to $u$ in $G$ has cost less than $\g$
due to the positive edge costs assumption. By inductive hypothesis, there
is a node in $T$ corresponding to that path to $u$ and, at the time this
node was added to $T$, it was also stored on the stack together with the
cost of the path to it.

As the stack ends up empty, at some later point $u$ had to be popped. Then 
$v$ is found to be a successor and the cost of the path from the 
root to $u$ popped from the stack, plus the cost of the edge $(u, v)$, 
is less than or equal to $\g$: $v$ is added to $T$ with $u$ as a parent, 
as the statement requires.
\end{proof}

\begin{theorem}
At the end of the algorithm, $R$ contains exactly the tight paths in $G$
for threshold~$\g$ that start at the root vertex.
\end{theorem}

\begin{proof}
Let $T$ be again the tree at the end of the algorithm.
A tight path in $G$ starting at the root vertex has cost $d \leq\g$ 
and, by Lemma~\ref{l:pathsGinT}, its last endpoint $v$ appears in $T$,
to which it is added at some point along the algorithm: at the same
point, $v$ enters the stack together with the cost $d$.

Since the stack is empty at the end, $(v, d)$ must be popped at some time
and, since the path is tight, any additional vertex either beyond $v$ or
before the root vertex incurs in a cost larger than $\g$: at that point,
the tight path is added to $R$.

Conversely, a path added to $R$ is given by a pair $(v, d)$ popped
from the stack, which it had to enter earlier. If that was upon
initialization, $v$ is the root vertex and the path has cost $0\leq\g$ 
in $G$. Otherwise, it was added due to a pair $(u, d')$ popped from
the stack and, by Lemma~\ref{l:pathsTinG}, from that point onward 
the path upwards from $v$ in $T$ corresponds to a path in $G$ of 
cost $d\leq\g$.

In either case, upon popping $(v, d)$ there is no potential extension 
with cost bounded by~$\g$ after $v$ or before the root as, otherwise, the
path would not have been added to $R$. Hence, paths in $R$ are tight paths 
in $G$.
\end{proof}

For later comparisons, we note here that the correctness proof does not
require acyclicity. In fact, going round several times through a cycle may
be crucial to reaching a tight cost. We will discuss below the potential 
risks of infinite loops if negative weights are allowed.

\subsection{Analysis}

An obvious lower bound for time is the total length of the tight 
paths to be found, as they must be written out. 
By Proposition~\ref{p:vcover}, every vertex belongs to at least one 
tight path, hence the total length of the tight paths is at least $|V|$.
As we explain next, upper bounding the total length 
of the tight paths in terms of only the size of the graph 
is not possible. Even realizing that $\g$ is part of the input, 
we may end up with huge quantities (on the combined lengths 
of $G$ and $\g$) of tight paths. Consider the case of 
integer-valued $\g$ for a constant-size graph, namely, 
the one depicted in Figure~\ref{fig:withcycles}. 
Consider all the integers on $t$ bits, from 0~to~$2^t - 1$,
for an arbitrary integer $t$. 
We can associate to each such integer a different path 
on that graph by following sequentially the bits, say, 
left to right, and looping on A-B-C-A when it is a 0 and 
on A-D-E-A when it is a 1, paying a cost of 4 per loop: 
by setting $\g = 4t$, this set of paths is exactly the set 
of tight paths that start at A, and there are $2^t = 2^{\g/4}$ of them. 
(This is actually a double exponential, noting that the 
value of $\g$ is exponential in its length.)

\looseness = -1
However, we can develop a time complexity analysis by taking
into the picture the final size of the bookkeeping tree constructed
for each root vertex.
The time complexity of the initialization is constant except for
traversing the predecessors of the root, which is bounded by the
graph size (or, more tightly, by its in-degree).
This cost will be dominated by that of the rest of the computation.
In the combined while/for loops, the internal operations can be 
made of constant cost so we need to bound the number of loops or,
equivalently, the number of times the inequality $d + d(u, v) \leq \g$
is evaluated. In total, the successful inequalities correspond one 
to one with parent links in the tree, but further inequalities are
tested: for every node in the tree, all the successors in $G$ must
be considered. We can bound the time by the size of the tree times
the number of vertices or, again more tightly, times the out-degree of $G$.

Once we have the tree ready, the upwards traversals in order to 
fill in the resulting set~$R$ only start at leaves, because proper 
subpaths of a tight path are not 
tight. Hence, the total length of the tight paths is the sum of lengths 
of paths from leaves to root in the final tree~$T$. The costliest case 
is when the tree is binary as, otherwise, upwards traversals are shorter;
then the total sum is $|T|\log|T|$. 

Whenever the out-degree of $G$ grows at most as $\log|V|$,
which covers many cases of interest, the traversals of the tree 
dominate the cost, which is, then, proportional to the total length
of the output.
In graphs of larger out-degree, the construction of the tree dominates,
for a time bound of $O(|V||T|)$.

Recall that the tree-constructing algorithm is to be run with each vertex
as root. Then, the total time bound is proportional to the maximum between
$
\sum_{u\in V} L(u) 
$,
where $L(u)$ is the sum of lengths of tight paths starting at $u$,
and $|V|^2|T|$ where $|T|$ is the size of the largest $T$ as the root 
traverses $G$.

\section{Tight pairs revisited}

We return now to the particular case of Section~\ref{s:acyclic}.
Of course, one can directly apply Algorithm~\ref{a:tightpaths};
but it turns out to be more efficient to simplify the algorithm,
obtaining also better running times than with the approach described
in Section~\ref{ss:ct}.

First note that, by acyclicity, 
we only need to consider simple paths; and, as paths 
joining the same pair have all the same weight, we can 
move to a scheme closer to standard depth-first search: 
we do not need to maintain the tree of paths and we can, 
instead, keep the usual set of visited vertices. In a 
first evolution, the stack still contains vertices together 
with distances from the start node. 
This variant is spelled out as Algorithm~\ref{a:tightpairsnaive}.

\begin{algorithm}
\DontPrintSemicolon
\KwData{graph $G$, root vertex in $G$, threshold $\g$}
\KwResult{set $R$ containing tight pairs in $G$ 
for threshold $\g$ where the first component 
of all pairs is that root}
\BlankLine
let $u$ be the closest predecessor of the root vertex in $G$ if any\;
let $d_0$ be $d(u, \mathrm{root})$ if the current root has some predecessor, 
$\infty$ otherwise\;
initialize both $R$ (results) and $S$ (seen vertices) to $\emptyset$\;
initialize an empty stack and push the pair $(\mathrm{root}, 0)$ into it\; 
\BlankLine
\While{\rm stack not empty}{
  let $(v, d)$ be the top or stack, pop it\;
  add $v$ to $S$\;
  \emph{mayextend} = \emph{False}\;
  \For{\rm each $u$ immediate successor of $v$ in $G$}{
    \If{$d + d(v, u) \leq \g$}
      {
      \If{$u\notin S$}
        {
        push $(u, d + d(v, u))$ into the stack\;
        }
      \emph{mayextend} = \emph{True}\;
      }
    }
\If{\rm not \emph{mayextend} and $d_0 + d > \g$}
   {add to $R$ the pair $(\mathrm{root}, v)$}
}
\BlankLine
\Return{$R$}
\BlankLine
\BlankLine
\caption{Tight pairs in a vertex-weighted graph 
starting at a root vertex}\label{a:tightpairsnaive}
\end{algorithm}

The Boolean value \emph{mayextend} is made explicit now because
of a little subtlety: maybe nothing is added to the stack because
all successors were already visited, but actually some of them
correspond to extensions of the path that still obey the bound.
This no longer can be distinguished simply by checking that something 
entered the stack, and it is thus clearer to make the Boolean flag
explicit.
An implementation of Algorithm~\ref{a:tightpairsnaive} is 
also available in the same public GitHub repository\footnote{See 
\url{https://github.com/balqui/tightpaths}, file 
\texttt{src/tightpair.py}.}.

One may wonder whether Lemma~\ref{l:acyclicpathsum} spares us
storing the distances in the stack. Indeed, we can find the cost
of any path in constant time by subtracting endpoint labels and
this allows for an even simpler algorithm that we report,
for completeness, as Algorithm~\ref{a:tightpairs}. 

\begin{algorithm}
\DontPrintSemicolon
\KwData{graph $G$, root vertex in $G$, threshold $\g$}
\KwResult{set $R$ containing tight pairs in $G$ 
for threshold $\g$ where the first component 
of all pairs is that root}
\BlankLine
let $u$ be the closest predecessor of the root vertex in $G$ if any\;
let $d_0$ be $d(u, \mathrm{root})$ if the current root has some predecessor, 
$\infty$ otherwise\;
initialize both $R$ (results) and $S$ (seen vertices) to $\emptyset$\;
initialize an empty stack and push the root into it\; 
\BlankLine
\While{\rm stack not empty}{
  let $v$ be the top or stack, pop it\;
  add $v$ to $S$\;
  \emph{mayextend} = \emph{False}\;
  \For{\rm each $u$ immediate successor of $v$ in $G$}{
    \If{$w(u) - w(\mathrm{root}) \leq \g$}
      {
      \If{$u\notin S$}
        {
        push $u$ into the stack\;
        }
      \emph{mayextend} = \emph{True}\;
      }
    }
\If{\rm not \emph{mayextend} and $d_0 + w(v) - w(\mathrm{root}) > \g$}
   {add to $R$ the pair $(\mathrm{root}), v)$}
}
\BlankLine
\Return{$R$}
\BlankLine
\BlankLine
\caption{Alternative approach to tight pairs 
in a vertex-weighted graph}\label{a:tightpairs}
\end{algorithm}

\begin{example}
\label{ex:acyclic13again}
Let us return to the toy graph in Figure~\ref{fig:acyclic13} and 
set again $\g = 0.9$. We start with the source node as root: 
both successors get stacked, and one of them brings in turn into 
the stack vertex $0.530$. As this vertex is popped, the Boolean flag
``mayextend'' remains false and we find the pair $(0.000, 0.530)$. 
The other successor of the root is then popped and the algorithm does 
nothing: $0.530$ is already visited but allows for extension,
so no further change is made and the run for this root is finished.
When each of these two successors is taken as root, only one 
further successor is reached, namely $0.530$ again in both cases, 
but the $d_0$ value coming from $0.000$ prevents paths from being 
reported. No successor enters the stack at all for root $0.530$ and, 
finally, upon starting with $1.115$ the sink is added to the stack 
and, as popped, yields the second and last tight pair $(1.115, 1.700)$.
Of course, nothing happens at the call on root $1.700$.
\end{example}

An implementation of the simplified Algorithm~\ref{a:tightpairs} is 
also available in the same public GitHub repository\footnote{See 
\url{https://github.com/balqui/tightpaths}, file 
\texttt{src/tightpair\_vw.py}, where ``\texttt{vw}'' stands for 
``vertex weights''.}.

\subsection{Correctness and Analysis}

The only difference between the last two algorithms is that, in one,
the distances are stored in the stack while, in the other, they get
computed on the fly from the vertex weights. We argue the correctness
of the last, simpler one but all the arguments work as well for the other.

\begin{theorem}
At the end of the algorithm, $R$ contains exactly the tight pairs in $G$
for threshold~$\g$ where the first element of the pair is the root vertex.
\end{theorem}

\begin{proof}
Suppose that a tight path exists in $G$ starting at the root vertex,
ending at vertex $v$ and, by Lemma~\ref{l:acyclicpathsum}, with cost 
$w(v) - w(\mathrm{root}) \leq \g$. By the standard argument according
to which DFS visits all reachable vertices, $v$ is visited at some
point: there, the path being tight, ``mayextend'' remains false and
$d_0$ suffices to add up to beyond $\g$, so that the pair is added to $R$.

Conversely, any pair added to $R$ with $v$ as second component is
added upon finding $v$ in the stack which implies that, first, it 
is reachable from the root vertex and, second, that 
$w(v) - w(\mathrm{root}) \leq \g$ because both conditions are
necessary in order to entering the stack. Moreover, at the time
it is added to $R$ the Boolean flag is still false, so that no 
extension beyond $v$ remained below the threshold, and also the
addition of a predecessor of the root would surpass the threshold.
Hence, it is a tight pair.
\end{proof}

As of cost, clearly the comparisons made inside the loops take 
constant time so that the depth-first search runs in time
$O(|V| + |E|)$; as the algorithm is run with each vertex
as root in turn, the total time is $O(|V|^2 + |V||E|)$.


\section{Applications and discussion}

We have stated the problem of finding tight paths in graphs
as well as a particular case where the path costs only depend 
on the endpoints. The application that motivated considering
this problem actually fits the particular case, but we consider
the general problem of independent interest.

\subsection{Basic antecedents}
\label{ss:ba}

As we have hinted at above, this particular case is motivated by a 
data analysis application. There, we have a fixed, finite set of items 
and a transactional dataset, which is simply a finite sequence of sets 
of items. Out of the dataset, certain sets of items, called ``closures'', 
are identified. It is known that the subset ordering structures them into
a lattice; see \cite{GanWil99}. 
The graph is the so-called Hasse graph of the lattice, while 
vertex weights correspond to so-called ``support'' measurements.
Depending on the dataset and on often applied bounds on those measurements, 
the graph size may be anything from very small to exponential in the
number of items.


By comparison with the acyclic graph setting of Section~\ref{s:acyclic}, 
in this application there is a slightly relevant difference, namely,
the labels of interest for the edges are ``confidences'', 
that is, quotients of the supports found as vertex labels, 
instead of differences. Both support and confidence are mere frequentist 
approximations to a probability (a conditional probability in the 
case of confidence), and the case can be reduced to our acyclic
graph case by, simply, log-scaling the weights in the vertices (and of 
course the threshold) so that the quotients take the form of 
differences. The basis of the logarithm is irrelevant, but a
practical warning is in order: comparing the float-valued 
weights requires care since there is a risk of precision-related
bad evaluations of equalities and inequalities.

Alternatively, we can leave weights and threshold unscaled and 
replace appropriately additions and differences by multiplications 
and quotients in the algorithms so as to obtain the desired tight pairs. 
Quotients are, then, such that we obtain probabilities in $[0, 1]$ 
(that is, we set the smaller endpoint weight as numerator) and
the tight pairs $(u,v)$ to be found are now to fulfill: $w(v)/w(u) \geq \g'$ 
but $w(v')/w(u) < \g'$ and $w(v)/w(u') < \g'$ for every predecessor $u'$ 
of $u$ and every successor $v'$ of $v$; here $\g'$ is the original 
conditional probability threshold which, scaled appropriately, leads 
to the $\g$ of the tight pairs problem. Inequalities become
reversed due to having the smaller value in the numerator instead
of the larger one: upon log-scaling, this also supposes a change of
sign so that weights are positive.
Casted in terms of confidence and partial implications, 
these conditions define the ``basic antecedents'' 
of~\cite{Bal2010LMCS}, appropriate to construct a basis of
partial implications. The closely related basis of ``representative
rules'' \cite{KryszIDA} is also connected to basic antecedents. 
If the vertex weights are integers, as in the 
supports case, this alternative additionally avoids logarithms and 
floats in implementations that may incur in risks of precision 
mismatches.
All these strategies allow us to apply tight pair algorithms 
in order to find basic antecedents with ease.
The already repeatedly mentioned open source repository 
that contains our programs implementing the algorithms 
in this paper provides as well examples of graphs obtained 
from datasets; there, the outcome of the tight pair algorithms 
is indeed the ``basic antecedent'' relation. Some of these
are used as examples in the next section.

\subsection{Comparing runtimes}

For the particular case of Section~\ref{s:acyclic},
we have now four algorithmic possibilities. 
We illustrate their comparative efficiencies 
in practice with some examples next.

\looseness = 1
The dataset NOW~\cite{NOW}\footnote{Initially Neogene of the Old World, 
see also: \url{https://en.wikipedia.org/wiki/Neogene_of_the_Old_World},
rebranded as a backronym New and Old World when American
sites were joined in.} 
contains information regarding paleontological excavation sites, 
including lists of Cenozoic land mammal species whose remains have been 
identified in each site, plus geographical coordinates 
and other data; four taxa per species are available 
(Order, Family, Genus and Species).
For our example, that we will call ``simplified NOW dataset'',
we keep just the Order taxon and consider a transactional dataset
where each transaction corresponds to an excavation site and contains
only the orders of the species whose remains have been found in the site.
From the data, a Hasse graph of the lattice of closed sets can be constructed, 
as mentioned in the previous Subsection~\ref{ss:ba}; the graph contains 
348 vertices.


\begin{figure}
    \centering
    \includegraphics[width=0.70\linewidth]{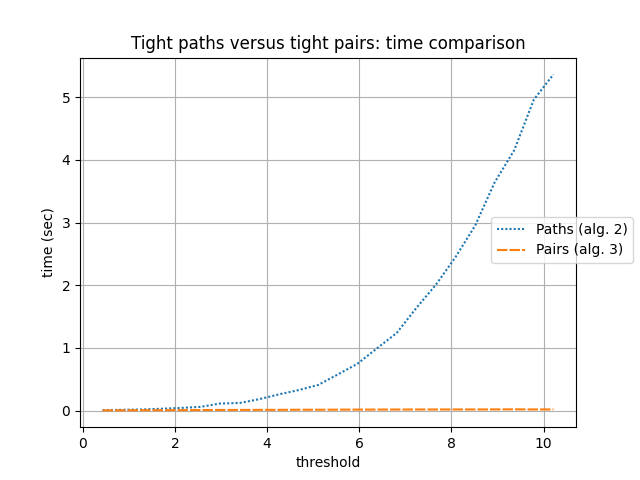}
    \caption{Comparison of running times for 
Algorithms \ref{a:tightpaths} and \ref{a:tightpairsnaive} 
on the Hasse graph of the closure space of the simplified NOW dataset.}
    \label{fig:timing23}
\end{figure}

We can see very clearly that, on this graph, Algorithm~\ref{a:tightpaths} is most often slower than 
Algorithm~\ref{a:tightpairsnaive}. Figure~\ref{fig:timing23}
shows (a piecewise linear interpolation of) the time required by each 
on 25 different, equally spaced threshold values. The reason of the
growing difference is as follows: as the threshold grows, on one hand
new pairs appear but, also, some appearing pairs replace pairs 
that are not tight anymore; more precisely, for thresholds 
below 9 there are less than 400 tight pairs and, beyond, the quantity
decreases further to below 200. However, as pairs are further 
and further apart, more and more different paths do connect them, 
so that the number of paths grows up to thousands at around 
threshold 1.5, then tens of thousands for thresholds beyond 4, 
and hundreds of thousands beyond threshold 8. 
Since Algorithm~\ref{a:tightpaths} explores all paths, 
its running time grows.

\begin{figure}
    \centering
    \includegraphics[width=0.70\linewidth]{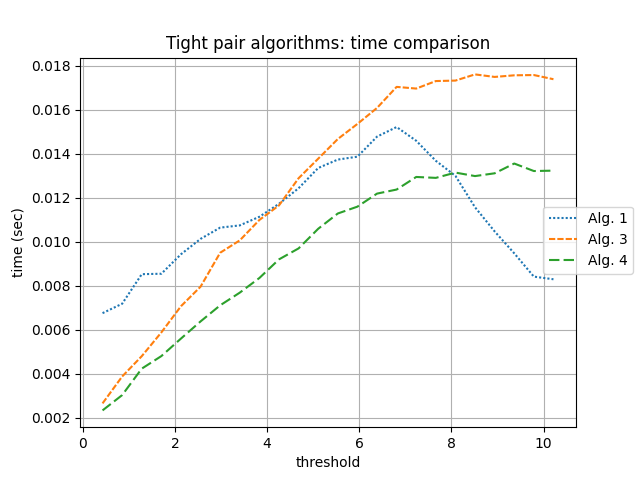}
    \caption{Comparison of running times for 
Algorithms \ref{a:slatt}, 
\ref{a:tightpairsnaive} and 
\ref{a:tightpairs} 
on the Hasse graph of the closure space of the simplified NOW dataset.}
    \label{fig:timing134}
\end{figure}

We believe that these considerations will remain applicable
to most input graphs.
Therefore, we recommend using Algorithm~\ref{a:tightpaths}
only for the general case and resort to one of the others 
for all cases where the conditions in Section~\ref{s:acyclic}
apply, such as the case of the simplified NOW dataset.
For this graph again, 
we see the same sort of graphic for the running times on
the same 25 threshold values in Figure~\ref{fig:timing134}.
Here the time is averaged over 5 separate, independent runs.
All times are below one-fiftieth of a second.
We see that the correspondence-based algorithm is only slightly
better on the higher end where the number of tight pairs
decreases whereas the number of paths connecting them grows.
We also see that the simplifications made to obtain 
Algorithm~\ref{a:tightpairs} do give it an edge. 
Additional tests suggest that the correspondence-based
implementation is faster on small graphs, probably due to the 
fast traversal of very small
dictionaries. Finally, Figure~\ref{fig:timing134m} shows
running times for a larger graph consisting of 2769 vertices,
also a Hasse graph but starting from a different transactional
dataset related to so-called ``market basket data''.
There, DFS-based algorithms clearly win at all thresholds 
explored, although the correspondence-based algorithm 
again improves its running times as the number of pairs starts
to decrease. Again the figure is obtained by averaging 5 runs 
on each of 25 equally spaced threshold values.

\begin{figure}
    \centering
    \includegraphics[width=0.70\linewidth]{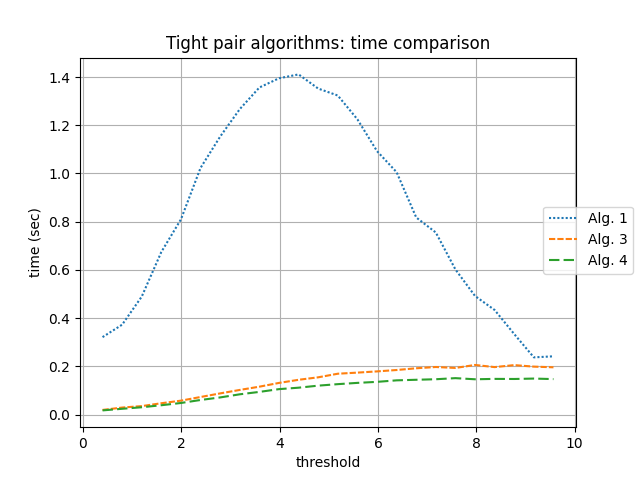}
    \caption{Comparison of running times for 
Algorithms \ref{a:slatt}, 
\ref{a:tightpairsnaive} and 
\ref{a:tightpairs} 
on the Hasse graph of the closure space of 
a ``market-basket style'' dataset.}
    \label{fig:timing134m}
\end{figure}

\subsection{Future work} 

Our Algorithm~\ref{a:tightpaths} is a first solution to the
general problem stated in this paper. The rest of our currently
completed work limits its scope to the specific variant motivated 
by our applications, and it is conceivable that better algorithms 
may be found for the general case.

For graphs where some edges have a negative cost, the first issue to
warn about is that the proof of Proposition~\ref{p:extendpath}
is not applicable anymore and, therefore, one should
clarify which one of the two equivalent statements of
that proposition is being considered for extension
into negative weights. Algorithm~\ref{a:tightpaths} 
is able to handle certain cases of negative costs
but in other cases it will loop forever.
The task of working out all the details and properties
of tight pairs in the presence of negative costs remains
as a topic for future research but, as a motivating 
question, we dare to state a specific conjecture: 
we currently believe that Algorithm~\ref{a:tightpaths}, 
applied as is to a graph with some negative or zero weights, 
will not finish if and only if there is in the graph a cycle 
of nonpositive total weight.






\bibliographystyle{plain}
\bibliography{b}

\end{document}